\documentclass[preprint,12pt]{elsarticle}
\usepackage[utf8]{inputenc}
\usepackage{amsmath}
\usepackage{empheq}
\usepackage{amssymb}
\usepackage{amsthm} 
\newtheorem{theorem}{Theorem}
\newtheorem{definition}[theorem]{Definition}
\newtheorem{lemma}[theorem]{Lemma}

\newtheorem{corollary}[theorem]{Corollary}
\usepackage{comment}
\usepackage{cases}

\usepackage{here}
\usepackage{caption}
\usepackage{algorithm}
\usepackage{algorithmicx}
\usepackage{algpseudocode}

\newcommand{\algname}[1]{{\rm{\texttt{#1}}}}

\makeatletter
\renewcommand{\fnum@algorithm}{\fname@algorithm}

\usepackage{algpseudocode}
\algnewcommand\algorithmicswitch{\textbf{switch}}
\algnewcommand\algorithmiccase{\textbf{case}}
\algdef{SE}[SWITCH]{Switch}{EndSwitch}[1]{\algorithmicswitch\ #1}{\algorithmicend\ \algorithmicswitch}
\algdef{SE}[CASE]{Case}{EndCase}[1]{\algorithmiccase\ #1}{\algorithmicend\ \algorithmiccase}
\algtext*{EndSwitch}
\algtext*{EndCase}

\usepackage{xcolor}
\newenvironment{revision}[1]{\color{black}}{\color{black}}

\begin{document}

\begin{frontmatter}
\title{Locally Defined Independence Systems on Graphs}
\author[1]{Yuki Amano}
\ead{ukiamano@kurims.kyoto-u.ac.jp}
\affiliation[1]{organization={Research Institute for Mathematical Sciences, Kyoto University},
            city={Kyoto},
            postcode={606-8502}, 
            country={Japan}}
\date{}


\begin{abstract}
    The maximization for the independence systems defined on graphs is a generalization of combinatorial optimization problems such as the maximum $b$-matching, the unweighted MAX-SAT, the matchoid, and the maximum timed matching problems. In this paper, we consider the problem under the local oracle model to investigate the global approximability of the problem by using the local approximability. We first analyze two simple algorithms \algname{FixedOrder} and \algname{Greedy} for the maximization under the model, which shows that they have no constant approximation ratio. Here algorithms \algname{FixedOrder} and \algname{Greedy} apply local oracles with fixed and greedy orders of vertices, respectively. We then propose two approximation algorithms for the $k$-degenerate graphs, whose approximation ratios are $\alpha +2k -2$ and $\alpha k$, where $\alpha$ is the approximation ratio of local oracles. The second one can be generalized to the hypergraph setting. We also propose an $(\alpha + k)$-approximation algorithm for bipartite graphs, in which the local independence systems in the one-side of vertices are $k$-systems with independence oracles.
\end{abstract}

\end{frontmatter}

\section{Introduction}
The maximization for independence system is one of the most fundamental combinatrial optimization problems \cite{welsh1976matroid, cook1997combinatorial, schrijver2003combinatorial}. An independence system is a pair $(E, \mathcal{I})$ of a finite set $E$ and a family $\mathcal{I} \subseteq 2^{E}$ that satisfies 
\begin{align}
    &\mathcal{I}\ \text{contains empty set, i.e.,}\ \emptyset \in \mathcal{I} \text{, and} \label{IS1}\\
    &J \in \mathcal{I}\ \text{implies}\ I \in \mathcal{I}\ \text{for any}\ I \subseteq J \subseteq E\label{IS2}.
\end{align}
Here a member $I$ in $\mathcal{I}$ is called an {\em independent set}. The property \eqref{IS2} means that $\mathcal{I}$ is downward closed.
The maximization problem for an independence system is to find an independent set with the maximum cardinality. This problem includes, as a special case, the maximum independent set of a graph, the maximum matching, the maximum set packing and the matroid (intersection) problems \cite{oxley2006matroid, welsh1976matroid, schrijver2003combinatorial}.

In this paper, we consider the following independence systems defined on graphs.
Let $G=(V, E)$ be a graph with a vertex set $V$ and an edge set $E$. For a vertex $v$ in $V$, let $E_v$ denotes the set of edges incident to $v$. In our problem setting, each vertex $v$ has a local independence system $(E_v, \mathcal{I}_v)$, i.e., $\mathcal{I}_v \subseteq 2^{E_v}$, and we consider the independence system $(E, \mathcal{I})$ defined by
\begin{align}
    \mathcal{I} = \{I \subseteq E \mid I \cap E_v \in \mathcal{I}_v\ \text{for all}\ v \in V \}.\label{locallyIS}
\end{align}
Namely, $(E, \mathcal{I})$ is obtained by concatenating local independence systems $(E_v, \mathcal{I}_v)$, and is called {\em an independence system defined on a graph} $G$. We assume without loss of generality that $\{e\} \in \mathcal{I}$ holds for any $e \in E$, since otherwise the underlying graph $G$ can be replaced by $G^\prime=(V, E\setminus \{e\})$.
In this paper, we consider the maximization problem for it, i.e., for a given graph $G=(V, E)$ with local independence systems $(E_v, \mathcal{I}_v)$, our problem is described as
\begin{align}
    \begin{split}
        \text{maximize}&\quad |I|\\
        \text{subject to}&\quad I \cap E_v \in \mathcal{I}_v\ \text{for all}\ v \in V\\
        &\quad I \subseteq E.
    \end{split}\label{problem}
\end{align}
Note that any independence system $(E, \mathcal{I})$ is viewed as an independence system defined on a star.  As an example of the problem \eqref{problem}, for $b:V \to \mathbb{Z}_+$, let $\mathcal{I}_v = \{I \subseteq E_v \mid |I| \leq b(v) \}$, i.e., $\mathcal{I}_v$ is $b(v)$-bounded for every $v$ in $V$. Then \eqref{locallyIS} denotes the family of $b$-matchings in $G$. In particular, if $b \equiv 1$, it corresponds to the family of matchings in $G$ \cite{lovasz1986matching}. More generally, if $(E_v, \mathcal{I}_v)$ is a matroid for every $v$ in $V$, then the family \eqref{locallyIS} is a so-called {\em matchoid} \cite{jenkyns1975matchoids}.
 
The maximum $b$-matching and matchoid problems are well-studied combinatorial optimization problems \cite{lovasz1986matching,jenkyns1975matchoids}. It is known that the maximum matchoid problem is NP-hard \cite{lovasz1978matroid} and it is $3/2$-approximable \cite{fujito19932, lee2013matroid}.

The unweighted maximum satisfiability problem (MAX-SAT) is another example of the problem \eqref{problem}. The unweighted MAX-SAT is the problem to find a variable assignment that maximizes the number of the satisfied clauses in a given conjunctive normal form (CNF). For a variable set $X$, let $\varphi = \bigwedge_{c \in C} c$ be a CNF, where $C$ denotes a set of clauses with variables in $X$. Define a bipartite graph $G = (V=X \cup C, E)$ and local independence systems $(E_v, \mathcal{I}_v)$ by

\begin{align*}
    E &= \{(x, c) \in X \times C \mid x\ \text{is contained in}\ c\}\\
    \mathcal{I}_v &= 
    \begin{cases}
    \{I \subseteq E_v \mid \text{either}\ I \subseteq C_+(v)\ \text{or}\ I \subseteq C_-(v) \} &\text{if}\ v \in X \\
    \{I \subseteq E_v \mid |I| \leq 1\} &\text{if}\ v \in C,
    \end{cases}
\end{align*}
where, for a variable $x \in X$, $C_+(x)$ and $C_-(x)$ respectively denote the sets of clauses in $C$ that contain literals $x$ and $\overline{x}$. Therefore the unweighted MAX-SAT is an example of the problem \eqref{problem}. As is well-known, the unweighted MAX-SAT is NP-hard and approximable in ratio $1.255$ \cite{avidorimproved}. 

(Edge-)temporal graphs \cite{kostakos2009temporal} were introduced to model dynamic network topologies where the edge set vary with time. Namely, a {\em temporal graph} is a graph $G=(V,E)$ with given time labels $L_e \subseteq T$ for all $e$ in $E$, where $T$ is a given finite set of time labels. For a temporal graph, a subset $M$ of $E$ is a {\em timed matching} if $L_e$ and $L_f$ are disjoint for any adjacent pair of edges $e$ and $f$ in $M$. By defining local independence systems $(E_v, \mathcal{I}_v)$ by $\mathcal{I}_v = \{I \subseteq E_v \mid L_e \cap L_f = \emptyset\ \text{for all}\ e,f \in I\}$, the maximum timed matching can be formulated as the problem \eqref{problem}.
It is known that the maximum timed matching problem is NP-hard even if the given graph $G$ is bipartite, and is solvable in polynomial time if $G$ is a tree and every $L_e$ is represented as an interval for a given order of time labels \cite{mandal20200}.

In this paper, we consider the problem \eqref{problem} by making use of local oracles $\mathcal{A}_v$ for each $v$ in $V$. 
The oracle $\mathcal{A}_v$ computes an $\alpha$-approximate independent set $I$ of $(F, \mathcal{I}_v[F])$ for a given set $F \subseteq E_v$, where $\mathcal{I}_v[F]$ is the restriction of $\mathcal{I}_v$ to $F$, i.e., $\mathcal{I}_v[F] = \{I \cap F \mid I \in \mathcal{I}_v\}$.
That is, the oracle $\mathcal{A}_v: 2^{E_v} \to 2^{E_v}$ satisfies
\begin{align}
    \mathcal{A}_v(F) &\in \mathcal{I}_v[F] \label{ALO1}\\
    \alpha\, |\mathcal{A}_v(F)| &\geq \max_{J \in \mathcal{I}_v[F]}|J|\label{ALO2}.
\end{align}
We call $\mathcal{A}_v$ an {\em $\alpha$-approximation local} oracle. It is also called an {\em exact local oracle} if $\alpha =1$. In this paper, we assume the monotonicity of $\mathcal{A}_v$, i.e., $|\mathcal{A}_v(S)| \leq |\mathcal{A}_v(T)|$ holds for the subsets $S \subseteq T \subseteq E_v$, which is a natural assumption on the oracle since it deals with independence system. We study this oracle model to investigate the global approximability of the problem \eqref{problem} by using the local approximability.
The oracle model was used in \cite{chekuri2004maximum} for the maximum coverage problem with group constraints, where the oracle model is regarded as a generalization of greedy algorithms. It is shown that the maximum coverage problem is $(\alpha + 1)$-approximable in this oracle model.
We remark here that a greedy algorithm for the maximum cardinality matching problem can be viewed as a $2$-approximation algorithm under the exact local oracle model. 

For the above reductions of the maximum matchoid problem and the unweighted MAX-SAT, their local maximization problems for $(E_v, \mathcal{I}_v)$ can be solved exactly, while their global maximization problems are NP-hard. This means that the problem \eqref{problem} seems to be intractable, even if exact local oracles are given. 

In this paper, we first propose two natural algorithms for the problem \eqref{problem}, where the first one applies local oracles $\mathcal{A}_v$ in the order of the vertices $v$ that is fixed in advance, while the second one applies local oracles in the greedy order of vertices $v_1, \dots, v_n$, where $n = |V|$ and 
\begin{align*}
    v_i \in \arg \max \{|\mathcal{A}_{v}(E_{v} \cap F^{(i)})| \mid v\in V \setminus\{v_1,\dots,v_{i-1}\}\}&& \text{for}\ i = 1, \dots, n.
\end{align*}
Here the subset $F^{(i)} \subseteq E$ is a set of available edges during the $i$-th iteration. 

We show that the first algorithm guarantees an approximation ratio $(\alpha + n-2)$, and the second algorithm guarantees an approximation ratio $\rho(\alpha , n)$, where $\rho$ is the function of $\alpha$ and $n$ defined as 
        \begin{empheq}[left={\rho(\alpha , n) = \empheqlbrace}]{align*}
            &\alpha + \frac{2\alpha-1}{2\alpha}(n-1) - \frac{1}{2}\quad &&\text{if}\ (\alpha-1)(n-1) \geq \alpha(\alpha + 1)\\
            &\alpha + \frac{\alpha}{\alpha+1}(n-1)\quad &&\text{if}\ \alpha \leq (\alpha-1)(n-1) < \alpha(\alpha + 1)\\
            &\frac{n}{2}\quad &&\text{if}\ (\alpha-1)(n-1) < \alpha.
        \end{empheq}
We also show that both of approximation ratios are {\it almost tight} for these algorithms. 

We then consider two subclasses of the problem \eqref{problem}.
We provide two approximation algorithms for the $k$-degenerate graphs, whose approximation ratios are $\alpha +2k -2$ and $\alpha k$. Here, a graph is {\em $k$-degenerate} if any subgraph has a vertex of degree at most $k$. This implies for example that the algorithms finds an $\alpha$-approximate independent set for the problem if a given graph is a tree. This is best possible, because the local maximization is not approximable with $c\ (< \alpha)$. We also show that the second algorithm can be generalized to the hypergraph setting.

We next provide an $(\alpha + k)$-approximation algorithm for the problem when a given graph is bipartite and local independence systems for one side are all $k$-systems with independence oracles. 
Here an independence system $(E, \mathcal{I})$ is called a {\em $k$-system} if for any subset $F \subseteq E$, any two maximal independent sets $I$ and $J$ in $\mathcal{I}[F]$ satisfy $k|I| \geq |J|$, and its independence oracle is to decide if a given subset $J \subseteq E$ belongs to $\mathcal{I}$ or not.

The rest of the paper is organized as follows. In Section 2, we describe two natural algorithms for the problem \eqref{problem} and analyze their approximation ratios. Section 3 provides approximation algorithms for the problem in which a given graph $G$ has bounded degeneracy. Section 4 also provides an approximation algorithm for the problem in which a given graph $G$ is bipartite, and all the local independence systems of the one side of vertices are $k$-systems. Section 5 defines independence systems defined on hypergraphs and generalizes algorithms to the hypergraph case.
\section{Local subpartitions and greedy algorithms}
    In this section, we first define local subpartitons of edges and analyze two natural algorithms for the problem \eqref{problem}. Local subpartitons of edges is used for constructing and analyzing approximation algorithms for the problem.
    
    \begin{definition}
        For a graph $G=(V, E)$ and subsets $P_v \subseteq E_v$ for all $v \in V$, the collection $\mathcal{P} = \{P_v \mid v \in V\}$ is called a {\em local subpartition} if $P_u \cap P_v = \emptyset$ for all distinct pairs $u$ and $v$ of vertices. The set $R = E \setminus (\bigcup_{P \in \mathcal{P}} P)$ is called the {\em residual} edge set of $\mathcal{P}$.
    \end{definition}
    By definition, for a local subpartition $\mathcal{P} = \{P_v \mid v \in V\}$, the set $P_v$ may be empty for some $v$ in $V$.
    Let $\mathcal{P} = \{P_v \mid v \in V\}$ be a local subpartition of $E$, and for every $v \in V$, let $I_v\in \mathcal{I}_v[P_v]$ be a subset of $P_v$ which is independent in $\mathcal{I}_v$. Then the union $I = \bigcup_{v \in V} I_v$ may {\it not} be an independent set of $\mathcal{I}$. If it is independent, we have the following lemma, where we recall that $\mathcal{A}_v\ (v \in V)$ denotes an $\alpha$-approximate local oracle.
        

    \begin{lemma}\label{lem:subpart}
        For a local subpartition $\mathcal{P} = \{P_v \mid v \in V\}$, let $R$ be the residual edge set of $\mathcal{P}$.
        If $I = \bigcup_{v \in V} \mathcal{A}_v(P_v)$ is an independent set in $\mathcal{I}$, then it guarantees an approximate ratio of $\alpha + \max_{J \in \mathcal{I}[R]}|J|/|I|$ for the problem \eqref{problem}. 
    \end{lemma}
    \begin{proof}
        For an independent set $K \in \mathcal{I}$, the set $I=\bigcup_{v \in V} \mathcal{A}_v(P_v)$ satisfies the following inequalities:
        \begin{align*}
            \alpha\, |I| &= \alpha\sum_{v \in V}|\mathcal{A}_v(P_v)| \geq \sum_{v \in V}|K \cap P_v| = |K| - |K \cap R|.
        \end{align*}
        This implies 
        \begin{align*}
            |K| \leq (\alpha + |K \cap R|/|I|)\, |I| \leq (\alpha + \max_{J \in \mathcal{I}[R]}|J|/|I|)\, |I|,
        \end{align*}
        which completes the proof.
    \end{proof}

    All the algorithms proposed in this paper construct local subpartitons of edges, and Lemma \ref{lem:subpart} is used to analyze their approximation ratio.
    
    Let us then see the following two simple algorithms which cannot guarantee any constant approximation ratio, even if exact local oracles are available. 
    The first one called \algname{FixedOrder} makes use of local oracles $\mathcal{A}_v$ in the order of the vertices $v \in V$ that is fixed in advance.
    \begin{algorithm}
        \caption{\rm{\algname{FixedOrder}}}
        \label{alg:greedy1}
        \begin{algorithmic}
        \State\Comment{\textwidth}{($v_1, \dots, v_n$) is a given vertex order.}
        \State $F := E$.\vspace{3pt}
        \For{$v = v_1, \dots, v_n$}
            \State $P_{v} := E_{v} \cap F$.\vspace{3pt}
            \State $R_{v} := \left(\left({\bigcup_{\substack{u \in V:\\ (u, {v}) \in \mathcal{A}_{v}(P_{v})}}}E_u \right) \setminus P_{v}\right) \cap F$.\vspace{3pt}
            \State $F := F \setminus (P_{v} \cup R_{v})$.\vspace{3pt}
        \EndFor
        \State $I := \bigcup_{v \in V}\mathcal{A}_v(P_v)$.\vspace{3pt}
        \State $R := \bigcup_{v \in V}R_v$.\vspace{3pt}
        \State Output $I$ and halt.
        \end{algorithmic}
    \end{algorithm}
    Algorithm \algname{FixedOrder} constructs a local subpartition $\mathcal{P} = \{ P_v \mid v \in V \}$ and the residual edge set $R$ of $\mathcal{P}$, which will be proven in the next theorem.
    In order to ensure that $I = \bigcup_{v \in V}\mathcal{A}_v(P_v)$ is independent in Lemma \ref{lem:subpart}, the algorithm maintains $F$ as a candidate edge set during the iteration. 
    The following theorem provides the approximation ratio of the algorithm, which is {\it almost tight} for the algorithm. In fact, it is tight if $\alpha$ is an integer.
    \begin{theorem}
        Algorithm \algname{FixedOrder} computes an $(\alpha + n-2)$-approximate solution for the maximization for the problem \eqref{problem} under the approximate local oracle model, and the approximation ratio of the algorithm is at least $\lfloor \alpha \rfloor+ n-2$.
    \end{theorem}
    \begin{proof}
        For a vertex $v \in V$, let $P_v$ and $R_v$ denote the edge sets computed in the algorithm, and let $I= \bigcup_{v \in V}\mathcal{A}_v(P_v)$ and $R = \bigcup_{v \in V} R_v$. Then we have $P_v \subseteq E_v$ for any $v \in V$ and $P_{v_i} \cap E_{v_j} = \emptyset$ if $j < i$. These imply that $\mathcal{P} = \{P_v \mid v \in V\}$ is a local subpartition of $E$. Note that $R_u \cap P_v = \emptyset$ for any $u$ and $v$ in $V$, which implies that $R \cap P_v = \emptyset$ for any $v$ in $V$. Since $P_{v_i} = E_{v_i} \setminus \bigcup_{j < i}(P_{v_j} \cup R_{v_j})$ also holds for any $v_i \in V$, $R$ is the residual edge set of $\mathcal{P}$. 
        Moreover, we have
        \begin{align}
            |R| \leq \sum_{v \in V}|R_v| \leq \sum_{v \in V}\sum_{\substack{u \in V:\\ (u, v) \in \mathcal{A}_v(P_v)}}|E_u \setminus P_v| \leq (n-2)|I|.
        \end{align}
        For any vertex $v_i \in V$, if an edge $(v_i, v_j)$ with $i > j$ is contained in $\mathcal{A}_{v_j}(P_{v_j})$, then $P_{v_l} \cap E_{v_i} = \emptyset$ holds for all $l > j$. This implies that $|I \cap E_v| \leq 1$ or $I \cap E_v = \mathcal{A}_{v}(P_{v})$ for all $v \in V$, which concludes that $I \in \mathcal{I}$. Therefore by Lemma \ref{lem:subpart}, $I$ guarantees an approximation ratio of $\alpha + n-2$ for the problem \eqref{problem}.
        
        We next show a lower bound of the approximation ratio of the algorithm.
        Define a graph $G=(V,E)$ by
        \begin{align*}
            V &= \{s, t\} \cup \{v_i \mid i = 1, \ldots, n-2\}\\
            E &= E_s \cup E_t,
        \end{align*}
        where
        \begin{align*}
            E_s &= \{(s, t)\} \cup \{(s, v_i) \mid i = 1, \ldots, \lfloor \alpha \rfloor - 1\}\ \text{and}\\
            E_t &= \{(s, t)\} \cup \{(t, v_i) \mid i = 1, \ldots, n-2\}.
        \end{align*}
        For every vertex $v \in V$, let $\mathcal{I}_v = 2^{E_v}$.
        By definition, $E$ is an optimal solution for this instance of the problem \eqref{problem}. On the other hand, if the algorithm first chooses a vertex $s$ and the local oracle $\mathcal{A}_{s}$ returns $\{(s, t)\}$, then the algorithm outputs $I = \{(s, t)\}$. Since $|E| = \lfloor \alpha \rfloor + n-2$ and $|I| = 1$, $\lfloor \alpha \rfloor + n-2$ is a lower bound of the approximation ratio of the algorithm.
    \end{proof}
    
    The second algorithm called \algname{Greedy} makes use of local oracles $\mathcal{A}_v$ in a greedy order of vertices $v_1, \dots, v_n$, where
    \begin{align*}
        v_i \in \arg \max \{|\mathcal{A}_{v}(E_{v} \cap F^{(i)})| \mid v\in V \setminus\{v_1,\dots,v_{i-1}\}\}&& \text{for}\ i = 1, \dots, n.
    \end{align*}
    Here the subset $F^{(i)} \subseteq E$ is the candidate edge set in the $i$-th round of \algname{Greedy}.
    \begin{algorithm}
        \caption{\rm{\algname{Greedy}}}
        \label{alg:greedy2}
        \begin{algorithmic}
        \State $F := E$.\vspace{3pt}
        \State $W := V$.\vspace{3pt}
        \While{$W \not = \emptyset$}
            \State $v \in \arg\max_{w \in W}|\mathcal{A}_w(E_w \cap F)|$.\vspace{3pt}
            \State $W := W \setminus \{v\}$.\vspace{3pt}
            \State $P_{v} := E_{v} \cap F$.\vspace{3pt}
            \State $R_{v} := \left(\left(\bigcup_{\substack{u \in V:\\ (u, v) \in \mathcal{A}_v(P_{v})}}E_u \right)\setminus P_{v} \right) \cap F$.\vspace{3pt}
            \State $F := F \setminus (P_{v} \cup R_{v})$.\vspace{3pt}
        \EndWhile
        \State $I := \bigcup_{v \in V}\mathcal{A}_v(P_v)$\vspace{3pt}
        \State $R := \bigcup_{v \in V}R_v$.\vspace{3pt}
        \State Output $I$ and halt.
        \end{algorithmic}
    \end{algorithm}
    For the analysis of \algname{Greedy}, we use the following lemma which is slightly different from Lemma \ref{lem:subpart}.
    
    \begin{lemma}\label{lem:subpart-2}
        Let $\mathcal{P} = \{P_v \mid v \in V\}$ be a local subpartition of $E$, and for the residual $R$ of $\mathcal{P}$, let $\{R_v \mid v \ \in V\}$ be a partition of $R$ that satisfies $R_v = \emptyset$ for all vertices $v$ with $P_v = \emptyset$.
        If $I = \bigcup_{v \in V} \mathcal{A}_v(P_v)$ is an independent set in $\mathcal{I}$, then it guarantees approximation ratio $\beta$ for the problem \eqref{problem}, where
        \begin{align*}
            \beta = \max_{v \in V: P_v \not = \emptyset }\ \max_{J \in \mathcal{I}[P_v \cup R_v]}\frac{|J|}{|\mathcal{A}_v(P_v)|}.
        \end{align*}
    \end{lemma}
    \begin{proof}
        For an independent set $K \in \mathcal{I}$, $I=\bigcup_{v \in V} \mathcal{A}_v(P_v)$ satisfies the following inequalities:
        \begin{align*}
            |K| \leq \sum_{v \in V} \max_{J \in \mathcal{I}[P_v \cup R_v]}|J|
            \leq \beta\sum_{v \in V: P_v \not = \emptyset} |\mathcal{A}_v(P_v)|
            = \beta\, |I|,
        \end{align*}
        which completes the proof.
    \end{proof}
    The following theorem provides the approximation ratio of \algname{Greedy}, which is {\it almost tight} for the algorithm. In fact, it is tight if $(\alpha-1)(n-1) < \alpha$.
    \begin{theorem}
        Algorithm \algname{Greedy} computes a $\rho(\alpha , n)$-approximate solution for the problem \eqref{problem}, where $\rho$ is the function of $\alpha$ and $n$ defined as 
        \begin{empheq}[left={\hspace{-5pt}\rho(\alpha , n) = \empheqlbrace}]{align}
            &\alpha + \frac{2\alpha-1}{2\alpha}(n-1) - \frac{1}{2} &&\hspace{-10pt}\text{if}\ (\alpha-1)(n-1) \geq \alpha(\alpha + 1)\label{lemeq:upperbounds1}\\
            &\alpha + \frac{\alpha}{\alpha+1}(n-1) &&\hspace{-10pt}\text{if}\ \alpha \leq (\alpha-1)(n-1) < \alpha(\alpha + 1)\label{lemeq:upperbounds2}\\
            &\frac{n}{2} &&\hspace{-10pt}\text{if}\ (\alpha-1)(n-1) < \alpha.\label{lemeq:upperbounds3}
        \end{empheq}
        Moreover, the approximation ratio is at least
        \begin{empheq}[left={\empheqlbrace}]{align*}
            &(\ref{lemeq:upperbounds1}) - \frac{\alpha}{2} &&\text{if}\ (\alpha-1)(n-1) \geq \alpha(\alpha + 1) &\\
            &(\ref{lemeq:upperbounds2}) - \left(\frac{3}{2} \alpha - \frac{1}{2} \right)&&\text{if}\ \alpha \leq (\alpha-1)(n-1) < \alpha(\alpha + 1)&\\
            &(\ref{lemeq:upperbounds3}) &&\text{if}\ (\alpha-1)(n-1) < \alpha.&
        \end{empheq}
    \end{theorem}
    \begin{proof}
        Let $v_1, \ldots, v_n$ be the order of vertices chosen in the while-loop of the algorithm.
        For $v \in V$, let $P_v$ and $R_v$ be the edge sets computed in the algorithm, and let $R = \bigcup_{v \in V}R_v$. 
        Then, similarly to Algorithm \algname{FixedOrder}, $\mathcal{P} = \{P_v \mid v \in V\}$ is a local subpartition of $E$ with the residual $R$, and $I$ is independent in $\mathcal{I}$.        
        Since $v$ is chosen from $\arg\max_{w \in W}|\mathcal{A}_w(E_w \cap F)|$ for each while-loop of Algorithm \algname{Greedy}, there is an index $s$ such that $(1)$ $P_{v_i} \not = \emptyset$ for all $i \leq s$ and $(2)$ $P_{v_i},R_{v_i} = \emptyset$ for all $i > s$. For $i \leq s$, define $F_i = P_{v_i} \cup R_{v_i}$. 
        
        Then, by Lemma \ref{lem:subpart-2}, it is enough to show that 
        \begin{align}
            \frac{\max_{J \in \mathcal{I}[F_i]}|J|}{|\mathcal{A}_{v_i}(P_{v_i})|} \leq \rho(\alpha, n)\ \text{for}\ i=1, \dots, s. \label{lemeq:upperratio}
        \end{align}
        Let $t$ be an index that attains the maximum in the left-hand side of (\ref{lemeq:upperratio}), and let $g_t = \max_{J \in \mathcal{I}[F_t]}|J|$.

        Let $H = (W, F_{t})$ be a graph induced by $F_t$, i.e., $W = \{v, u \in V \mid (v, u) \in F_t\}$. Let $\nu = |W|$ and $x = |\mathcal{A}_{v_t}(P_{v_t})|$. Note that
        \begin{align}
            F_{t} \subseteq E_{v_{t}} \cup \bigcup_{u \in W:\ (u, v_{t})\in \mathcal{A}_v(P_{v_t})}E_u, \label{lemeq:edges}
        \end{align}
        \begin{align}
            \max_{J \in \mathcal{I}[F_t \cap E_v]}|J| \leq \lfloor\alpha\, |\mathcal{A}_v(F_t \cap E_v)|\rfloor \leq \lfloor\alpha\, x\rfloor&& \text{for}\ v \in W, \label{lemeq:greedychoice}
        \end{align}
        and $v_t$ is chosen greedily. In order to estimate $g_t$, we consider subgraphs of $H$ with degree at most $\lfloor\alpha x\rfloor$.
        We separately deal with the following four cases.
        \begin{flalign*}
            \hspace{20pt}
            \text{Case 1:}&\ (\alpha -1)(n-1) < \alpha.&\\
            \text{Case 2:}&\ (\alpha -1)(n-1) \geq \alpha\ \text{and}\ \alpha x \geq \nu -1. &\\
            \text{Case 3:}&\ (\alpha -1)(n-1) \geq \alpha\ \text{and}\ \nu-x-1 \leq \alpha x < \nu-1.&\\
            \text{Case 4:}&\ (\alpha -1)(n-1) \geq \alpha\ \text{and}\ \alpha x < \nu-x-1.&
        \end{flalign*}
        Case 1 is to prove \eqref{lemeq:upperbounds3}, while Cases 2,3, and 4 are to prove \eqref{lemeq:upperbounds1} and \eqref{lemeq:upperbounds2}.

        \textbf{Case 1.}
        Note that $(\alpha -1)(n-1) < \alpha$ is equivalent to the condition $\alpha(n-2) < n-1$. Thus if $x < \nu -1\ (\leq n-1)$, then we have $\lfloor \alpha x \rfloor = x$, and  hence it follows from \eqref{lemeq:edges} and \eqref{lemeq:greedychoice} that 
        \begin{align*}
            g_t \leq x(\nu-x-1) + \frac{1}{2}((x+1)x - x(\nu-x-1)) = \frac{\nu x}{2},
        \end{align*}
        which implies that 
        \begin{align*}
            \frac{g_t}{x} \leq \frac{\nu}{2} \leq \frac{n}{2}.
        \end{align*}
        On the other hand, if $x \geq \nu -1$, then we have
        \begin{align*}
            \frac{g_t}{x} \leq \frac{\frac{1}{2}\nu(\nu - 1)}{x} \leq \frac{\nu}{2} \leq \frac{n}{2}.
        \end{align*}
        In either case, (\ref{lemeq:upperbounds3}) is proved.
        
        \medskip
        
        \textbf{Case 2.}
        Since $\alpha x \geq \nu -1$, \eqref{lemeq:edges} implies that
        \begin{align*}
            g_t \leq (x+1)(\nu-x-1) + \frac{1}{2}(x+1)x \leq -\frac{1}{2}x^2 + \left(\nu - \frac{3}{2}\right)x + \nu-1.
        \end{align*}
        Thus,
        \begin{align*}
            \frac{g_t}{x} &\leq -\frac{1}{2}x + \left(\nu - \frac{3}{2}\right) + \frac{\nu-1}{x}\\
            &\leq \alpha + \frac{2\alpha -1}{2\alpha}(\nu-1) -\frac{1}{2}\\
            &\leq \alpha + \frac{2\alpha -1}{2\alpha}(n-1) -\frac{1}{2}, 
        \end{align*}
        where the second inequality is obtained by substituting $x = \frac{\nu -1}{\alpha}$.
        
        \medskip
        
        \textbf{Case 3.}
        Let $L$ be an independent set in $\mathcal{I}[F_t]$. Then we have $|L \cap E_v| \leq \lfloor \alpha x\rfloor$ for any $v \in V$, since $x = \mathcal{A}_{v_t}(P_{v_t})$ and local oracle $\mathcal{A}_v$ is $\alpha$-approximate. Thus if $\alpha x \geq x + 1$, it follows from (\ref{lemeq:edges}) that
        \begin{align*}
            g_t &\leq (x+1)(\nu-x-1) + \frac{1}{2}(x+1)(\lfloor \alpha x\rfloor - (\nu-x-1))\\ 
            &\leq\frac{1}{2}\left\{ (\alpha -1)x^2 + (\alpha + \nu -2)x + \nu -1\right\}, 
        \end{align*}
        which implies that
        \begin{align}
            \frac{g_t}{x} &\leq \frac{1}{2}\left( (\alpha -1)x + \alpha + \nu -2 + \frac{\nu -1}{x}\right).\label{lemeq:gtx22}
        \end{align}
        Since $\frac{\nu -1}{\alpha+1}\leq x \leq \frac{\nu -1}{\alpha}$, the right hand side of (\ref{lemeq:gtx22}) take take maximum at $x = \frac{\nu -1}{\alpha+1}$ or $\frac{\nu -1}{\alpha}$.  Namely, we have
        \begin{align}
            \frac{g_t}{x} \leq \max\left\{
                \alpha + \frac{\alpha}{\alpha+1}(n-1),\
                \alpha + \frac{2\alpha-1}{2\alpha}(n-1) - \frac{1}{2}
            \right\}. \label{lemeq:gtx}
        \end{align}
        On the other hand, if $\alpha x < x + 1$, i.e., $\lfloor \alpha x\rfloor = x$, then we have
        \begin{align*}
            g_t \leq x(\nu-x-1) + \frac{1}{2}((x+1)x - x(\nu-x-1)) = \frac{\nu x}{2}, 
        \end{align*}
        which implies that $\frac{g_t}{x} \leq \frac{\nu}{2} \leq \frac{n}{2}$. Since $\frac{n}{2} \leq \alpha + \frac{2\alpha-1}{2\alpha}(n-1) - \frac{1}{2}$ and $\frac{n}{2} \leq \alpha + \frac{\alpha}{\alpha+1}(n-1)$ for $\alpha \geq 1$, (\ref{lemeq:gtx}) is obtained in Case 3.

        \medskip
        
        \textbf{Case 4.}
        In this case, we have $g_t \leq (x+1)\lfloor \alpha x\rfloor$ and 
        \begin{align*}
            \frac{g_t}{x} \leq \alpha + \alpha x \leq \alpha + \frac{\alpha}{\alpha+1}(n-1),
        \end{align*}
        where the second inequality is obtained from $\alpha + \alpha x$ by substituting $x = \frac{\nu -1}{\alpha+1}$.
        
        \medskip
        
        In summarizing Cases 2, 3, and 4, we obtain inequality (\ref{lemeq:gtx}). Since 
        \begin{align*}
            \alpha + \frac{2\alpha-1}{2\alpha}(n-1) - \frac{1}{2} \geq \alpha + \frac{\alpha}{\alpha+1}(n-1)
        \end{align*}
        if and only if $(\alpha -1)(n-1) \geq \alpha(\alpha + 1)$, (\ref{lemeq:upperbounds1}) and (\ref{lemeq:upperbounds2}) are obtained.
        
        We next show lower bounds of the approximation ratio of the algorithm. 
        Let $(E, \mathcal{I})$ be an independence system on a complete graph $K_n = (V, E)$ with $I = 2^E$, and we assume that $\mathcal{A}_v(E_v) = E_v$ for some $v \in V$. Then Algorithm \algname{Greedy} chooses a vertex $v \in V$ with $\mathcal{A}_{v}(E_{v}) = E_v$ and sets $P_v = E_v$ and $R_v = E \setminus E_v$ at the first iteration of the while-loop. Since $F = \emptyset$ after the first iteration of the while-loop, the algorithm outputs $I = E_v$, and this implies that $\frac{n}{2}$ is a lower bound of the approximation ratio of the algorithm. Note that (\ref{lemeq:upperbounds3}) is tight if $(\alpha -1)(n-1) < \alpha$. Moreover, if $(\alpha-1)(n-1) < \alpha(\alpha + 1)$ holds, we have
        \begin{align*}
            (\ref{lemeq:upperbounds2}) - (\ref{lemeq:upperbounds3}) = \alpha - \frac{1}{2} + \frac{1}{2(\alpha+1)}(\alpha - 1)(n-1) \leq \frac{3}{2}\alpha - \frac{1}{2},
        \end{align*}
        which proves the tightness of \eqref{lemeq:upperbounds2}.
        We consider a lower bound of the algorithm if $(\alpha-1)(n-1) \geq \alpha(\alpha + 1)$ holds.
        Define a graph $G = (V=U \cup W, E)$ by 
        \begin{align*}
            U &= \{v_i \mid i = 1, \dots, \left\lceil \frac{n-1}{\alpha} \right\rceil + 1\}\\
            W &= \{v_i \mid i = \left\lceil \frac{n-1}{\alpha} \right\rceil + 2, \dots, n\}\\
            E_v &= \begin{cases}
                \{(v, u) \mid u \in V \setminus \{v\}\} & \text{if}\ v \in U\\
                \{(v, u) \mid u \in U\} & \text{if}\ v \in W.
            \end{cases}
        \end{align*}
        Let $(E, \mathcal{I})$ be an independence system on $G$ with $E \in \mathcal{I}$, and we assume that every local oracle returns an independent set of size $\left\lceil \frac{n-1}{\alpha} \right\rceil$. Note that $\alpha \left\lceil \frac{n-1}{\alpha} \right\rceil \geq n-1$ and $(\alpha - 1) \left\lceil \frac{n-1}{\alpha} \right\rceil \geq 1$ hold, since $(\alpha-1)(n-1) \geq \alpha(\alpha + 1)$. If Algorithm \algname{Greedy} chooses $v_1 \in U$ with $\mathcal{A}_{v_1}(E_{v_1}) = \{(v_1, u) \mid u \in U \setminus \{v_1\}\}$ at the first iteration of the while-loop, then $R_{v_1} = E \setminus E_{v_1}$ holds and the algorithm outputs $I = \mathcal{A}_{v_1}(E_{v_1})$.
        Similarly to the analysis of an upper bound of Case 2, we have
        \begin{align*}
            \frac{g_1}{\left\lceil \frac{n-1}{\alpha} \right\rceil} &= -\frac{1}{2}\left\lceil \frac{n-1}{\alpha} \right\rceil + \left(n - \frac{3}{2}\right) + \frac{n-1}{\left\lceil \frac{n-1}{\alpha} \right\rceil}\\
            &\geq -\frac{1}{2}\left(\frac{n-1}{\alpha}+ 1\right) + \left(n - \frac{3}{2}\right) + \frac{n-1}{ \frac{n-1}{\alpha} +1}\\
            &= \left(\alpha + \frac{2\alpha-1}{2\alpha}(n-1) - \frac{1}{2}\right) -\left(\frac{1}{2} + \frac{\alpha^2}{n-1 + \alpha}\right)\\
            &\geq \left(\alpha + \frac{2\alpha-1}{2\alpha}(n-1) - \frac{1}{2}\right) -\frac{\alpha}{2},
        \end{align*}
        where the first and second inequalities follow from $\left\lceil \frac{n-1}{\alpha} \right\rceil < \frac{n-1}{\alpha}+ 1$ and $(\alpha-1)(n-1) \geq \alpha(\alpha + 1)$, respectively. This proves the tightness of \eqref{lemeq:upperbounds1}.
    \end{proof}

\section{Approximation algorithms based on local oracles}    
    In this section, we present approximation algorithms for the problem \eqref{problem}.
    Our first algorithm called \algname{OrderedApprox} makes use of a linear order $\prec$ of vertices $V$. Different from \algname{FixedOreder} and \algname{Greedy} in Section 2, the algorithm tries to minimize the size of $R$. In fact, we have $R = \emptyset$ if $G$ is a tree. More precisely, for each $v \in V$, we initialize $P_v = D_v$, where $D_v$ is the set of downward edges of $v$, i.e., $D_v = \{(u, v) \in E_v \mid u \prec v\}$, and compute $\mathcal{A}_v(P_v)$ and $\mathcal{A}_v(P_v \cup \{e\})$ for each $e \in U_v$, where $U_v$ is the set of upward edges of $v$, i.e., $U_v = \{(u, v) \in E_v \mid v \prec u\}$. Based on these outputs of local oracle and their values, we consider four cases, each of which we update $P_v$ and construct $I_v = \mathcal{A}_v(P_v)$ and $R_v$ accordingly. Here $P_v$ and $R = \bigcup_{v \in V} R_v$ correspond to those in Lemma \ref{lem:subpart}. We then modify $P_w$ for vertices $w$ with $v \prec w$, in such a way that the set $I = \mathcal{A}_w(P_w) \cup \bigcup_{u \in V: u \preceq v} I_v$ is an independent set in $\mathcal{I}$.
    The second algorithm first decomposes edge set $E$ into forests $E_1, \ldots ,E_\gamma$, for each forest $E_i$, applies \algname{OrderedApprox} to compute an independent set $I_i$, and chooses a maximum independent set among them.

    We first define the upward and downward edge sets with respect to a linear order.
    \begin{definition}\label{def:order}
        For a graph $G = (V, E)$, let $\prec$ be a linear order of vertices $V$.
        For a vertex $v$, let $U_{v} = \{(v, w) \in E_v \mid v \prec w\}$ and $D_{v} = \{(v, w) \in E_v \mid w \prec v\}$ be the sets of {\em upward} and {\em downward} edges incident to $v$, respectively. We define the {\em width} of $G$ (with respect to a linear order $\prec$) as $\max_{v \in V}|U_{v}|$.
    \end{definition}
    
    The minimum width of the graph $G=(V, E)$ among all linear order of vertices $V$ is called the {\em degeneracy} of $G$, and $G$ is called {\em $k$-degenerate} if $k$ is at least the minimum width of $G$ \cite{lick_white_1970, freuder1982sufficient, matula1983smallest}. Note that a linear order of vertices $V$ certifying that $G$ is $k$-degenerate can be obtained by repeatedly choosing vertices with minimum degree in the remaining graph, i.e.,
    \begin{align*}
        v_i \in \arg\min \{\deg_{G[V \setminus \{v_1, \dots, v_{i-1}\}]}(v) \mid v \in V \setminus \{v_1, \dots, v_{i-1}\}\} && \text{for}\ i = 1, \dots, n,
    \end{align*}
    where $\deg_H(v)$ denotes the degree of vertex $v$ in the graph $H$ and $G[W]$ denotes the subgraph of $G$ induced by a vertex subset $W$.
    Therefore, such a liner order can be computed in linear time.
    It is also known that the degeneracy of a graph is at most the tree-width of it.
    \begin{algorithm}
        \caption{\rm{\algname{OrderedApprox}}$(G, \mathcal{I}, \prec)$}
        \label{alg:ordered}
        \begin{algorithmic}[1]
            \Require An independence system $(E, \mathcal{I})$ defined on a graph $G = (V, E)$ and a linear order $v_1 \prec v_2 \prec \ldots \prec v_n$ of vertices $V$.
            \Ensure An independent set in $\mathcal{I}$.
            \vspace{1pt}
            \State $X := \emptyset$.
            \State $P_v := D_v$ for $v \in V$.
            \For{$i = 1, \ldots, n$}\label{algline:firstfor}
            	\State $B_{v_i} := \{e \in U_{v_i} \mid e \in \mathcal{A}_{v_i}(P_{v_i} \cup \{e\}),\ |\mathcal{A}_{v_i}(P_{v_i} \cup \{e\})| > |\mathcal{A}_{v_i}(P_{v_i})|\}$.
            	\If{$U_{v_i}= \emptyset$}
    				\State $I_{v_i} := \mathcal{A}_{v_i}(P_{v_i})$. \label{algline:update_I-1}
    				\State $R_{v_i}:=\emptyset$. \label{algline:constructR-1}
            	\ElsIf{$U_{v_i} \not = \emptyset$, $P_{v_i} \not = \emptyset$, and $B_{v_i} = \emptyset$}
				    \State Choose an edge $e=(v_i, v_j) \in U_{v_i}$ arbitrarily.
				    \If{$e \not \in \mathcal{A}_{v_i}(P_{v_i} \cup\{e\})$}
				        \State $I_{v_i} := \mathcal{A}_{v_i}(P_{v_i} \cup \{e\})$. \label{algline:update_I-2a}
				    \Else{ (i.e., $|\mathcal{A}_{v_i}(P_{v_i} \cup \{e\})| = |\mathcal{A}_{v_i}(P_{v_i})|$)}
				        \State $I_{v_i} := \mathcal{A}_{v_i}(P_{v_i})$.\Comment{16.5em}{$|I_{v_i}| = |\mathcal{A}_{v_i}(P_{v_i}\cup \{e\})|$} \label{algline:update_I-2b}
				    \EndIf
    				\State $P_{v_i} := P_{v_i} \cup \{e\}$.\label{algline:update_P-1}
    				\State $R_{v_i} := U_{v_i} \setminus \{e\}$. \label{algline:constructR-2}
    				\State $P_{v_j} := P_{v_j} \setminus \{e\}$.\label{algline:update_P-2}
    			\ElsIf{$U_{v_i} \not = \emptyset$, $P_{v_i} \not = \emptyset$, and $B_{v_i} \not = \emptyset$}
				    \State Choose an edge $b \in B_{v_i}$ arbitrarily.
					\State $I_{v_i} := \mathcal{A}_{v_i}(P_{v_i} \cup \{b\}) \setminus \{b\}$. \Comment{16.5em}{$I_{v_i} \subseteq P_{v_i}$ and $|I_{v_i}| = |\mathcal{A}_{v_i}(P_{v_i})|$} \label{algline:update_I-3}
					\State $R_{v_i} := U_{v_i} \setminus \{b\}$.\label{algline:constructR-3}
    			\Else{ (i.e., $U_{v_i} \not = \emptyset$ and $P_{v_i} = \emptyset$)}
    				\State $I_{v_i} := \emptyset$. \label{algline:update_I-4}
    				\State $R_{v_i}:=\emptyset$. \label{algline:constructR-4} \Comment{16.5em}{$R_{v_i}$ might be updated in line \ref{algline:update_R}}
    				\State $X := X \cup \{v_i\}$.
                \EndIf
				\For{$(v_l, v_i) \in I_{v_i}$}
				    \If{$v_l \in X$}
				        \State $R_{v_l} := \{(v_l, v_j) \in U_{v_l} \mid j>i\}$.\label{algline:update_R}
				        \State $X := X \setminus \{v_l\}$.
				    \EndIf
				\EndFor
				\algstore{break1}
		\end{algorithmic}
    \end{algorithm}
    \begin{algorithm}
        \begin{algorithmic}[1]
				\algrestore{break1}
				\For{$j = i+1, \ldots, n$}
				\State $P_{v_j} := P_{v_j} \setminus (\bigcup_{l \leq i} R_{v_l})$.\label{algline:update_P-3}
				\EndFor
            \EndFor
            \State Output $I= \bigcup_{v \in V} I_v$ and halt.
        \end{algorithmic}
    \end{algorithm}
    
    Algorithm \algname{OrderedApprox} first initializes $P_v = D_v$ for all $v \in V$ and for each $i$-th iteration of the for-loop, computes an edge set $B_{v_i} \subseteq U_{v_i}$ by
    \begin{align*}
        B_{v_i} := \{e \in U_{v_i} \mid e \in \mathcal{A}_{v_i}(P_{v_i} \cup \{e\}),\ |\mathcal{A}_{v_i}(P_{v_i} \cup \{e\})| > |\mathcal{A}_{v_i}(P_{v_i})|\}.
    \end{align*} 
    It separately treats the following four cases as in the description of the algorithm.
    \begin{flalign*}
        \hspace{20pt}\text{Case 1:}&\ U_{v_i} = \emptyset &\\
        \text{Case 2:}&\ U_{v_i} \not= \emptyset,\ P_{v_i} \not= \emptyset,\ \text{and}\ B_{v_i} = \emptyset&\\
        \text{Case 3:}&\ U_{v_i} \not= \emptyset,\ P_{v_i} \not= \emptyset,\ \text{and}\ B_{v_i} \not= \emptyset&\\
        \text{Case 4:}&\ U_{v_i} \not= \emptyset\ \text{and}\ P_{v_i} = \emptyset&
    \end{flalign*}
    For all the cases, we show the following two lemmas.
    \begin{lemma}\label{lem:conditionP}
    Algorithm \algname{OrderedApprox} satisfies the following three conditions
    \begin{flalign*}
        \hspace{10pt} \rm{({\romannumeral 1})}&\ I_{v_j} \subseteq P_{v_j} \subseteq E_{v_j}\ \text{with}\ |I_{v_j}| = |\mathcal{A}_{v_j}(P_{v_j})|\ \text{for all }\ j \leq i,&\\
        \rm{({\romannumeral 2})}&\ R_{v_j} \subseteq U_{v_j}\ \text{for all }\ j \leq i,\ \text{and}&\\
        \rm{({\romannumeral 3})}&\ \mathcal{P} = \{P_v \mid v \in V\}\ \text{is a local subpartition of}\ E\ \text{with residual}\ R=\bigcup_{j \leq i}R_{v_j}&
    \end{flalign*}
    at the end of the $i$-th iteration of the for-loop.
    \end{lemma}
    \begin{proof}
        Since $I_{v_j}$ and $P_{v_j}$ are never modified after the $j$-th iteration of the for-loop in Line \ref{algline:firstfor}, it is enough to show that $I_{v_i} \subseteq P_{v_i} \subseteq E_{v_i}$ and $|I_{v_i}| = |\mathcal{A}_{v_i}(P_{v_i})|$ at the end of the $i$-th iteration of the for-loop to prove Condition \rm{({\romannumeral 1})}. We can see that $P_{v_i}$ is initialized to $D_{v_i}$, and $P_{v_i}$ is modified in Lines \ref{algline:update_P-1}, \ref{algline:update_P-2}, and \ref{algline:update_P-3}. In Lines \ref{algline:update_P-2} and \ref{algline:update_P-3}, no edge is added to $P_{v_i}$, and in Line \ref{algline:update_P-1}, some edge $e \in U_{v_i}$ is added to $P_{v_i}$.
        These show that $P_{v_i} \subseteq E_{v_i}$. Moreover, since $I_{v_i}$ is constructed in Lines \ref{algline:update_I-1}, \ref{algline:update_I-2a}, \ref{algline:update_I-2b}, \ref{algline:update_I-3}, and \ref{algline:update_I-4}, we have $I_{v_i} \subseteq P_{v_i}$ and $|I_{v_i}| = |\mathcal{A}_{v_i}(P_{v_i})|$, which implies Condition \rm{({\romannumeral 1})}.
        
        Since $R_{v_j}$ is never modified after the $j$-th iteration for all the cases except for Case 4, Condition \rm{({\romannumeral 2})} is satisfied if Case 4 is not satisfied in the $j$-th iteration of the for-loop. On the other hand, if Case 4 is satisfied in the $j$-th iteration of the for-loop, then $R_{v_j}$ might be updated in Line \ref{algline:update_R}, which again satisfies $R_{v_j} \subseteq U_{v_j}$.        
        This implies Condition \rm{({\romannumeral 2})}.
        
        Let us finally show Condition \rm{({\romannumeral 3})}. Before the first iteration of the for-loop, $\mathcal{P} = \{P_v = D_v \mid v \in V\}$ is a local subpartition of $E$ with the residual $R = \emptyset$. Assuming that Condition \rm{({\romannumeral 3})} is satisfied in the beginning of the $i$-th iteration, we show that Condition \rm{({\romannumeral 3})} is satisfied at the end of the $i$-th iteration. Note that $P_{v_i}$ is modified in Lines \ref{algline:update_P-1}, \ref{algline:update_P-2}, and \ref{algline:update_P-3}. In Lines \ref{algline:update_P-2} and \ref{algline:update_P-3}, no edge is added to $P_{v_j}$ for any $j$. If some edge $e=(v_i, v_j) \in U_{v_i}$ is added to $P_{v_i}$ in Line \ref{algline:update_P-1}, $e$ is removed from $P_{v_j}$. These imply that $\mathcal{P} = \{P_v \mid v\in V\}$ is a subpartition of $E$. Moreover, 
        in any iteration, no edge is added to $\bigcup_{P_{v} \in \mathcal{P}} P_v$, and when $R_v$ is constructed in Lines \ref{algline:constructR-1}, \ref{algline:constructR-2}, \ref{algline:constructR-3}, \ref{algline:constructR-4}, and \ref{algline:update_R}, all the edges in such an $R_v$ are deleted from the corresponding $P_{v_j}$ in Line \ref{algline:update_P-3}. Thus Condition \rm{({\romannumeral 3})} is satisfied at the end of the $i$-th iteration, which completes the proof of the lemma.
    \end{proof}
    \begin{lemma}\label{lem:independence}
    Algorithm \algname{OrderedApprox} satisfies that
    \begin{align}
        I \cup \bigcup_{j < i}I_{v_j} \in \mathcal{I}\ \text{for any independent set}\ I \in \mathcal{I}[P_{v_i}] \label{lemeq:independence}
    \end{align}
    in the beginning of the $i$-th iteration of the for-loop.
    \end{lemma}
    \begin{proof}
        For an index $j$, consider the end of the $j$-th iteration of the for-loop in Line \ref{algline:firstfor}. We have $I_{v_j} \cup \{e\} \in \mathcal{I}_{v_j}$ for any $e \in U_{v_j} \setminus (P_{v_j} \cup R_{v_j})$. If the $j$-th iteration of the for-loop falls into Cases 1, 2, or 3, then $U_{v_j} \cap (\bigcup_{k > j} P_{v_k}) = U_{v_j} \setminus (P_{v_j} \cup R_{v_j})$ contains at most one edge. Otherwise (i.e., Case 4), we have $U_ {v_j} \subseteq \bigcup_{k > j} P_{v_k}$, and $R_{v_j}$ will be updated to $\{(v_j, v_l) \in U_{v_j} \mid l > k\}$ once $(v_j, v_k) \in U_{v_j}$ is chosen by $I_{v_k}$ in the $k$-th iteration for some $k > j$. Therefore, (\ref{lemeq:independence}) is satisfied in the beginning of the $i$-th iteration of the for-loop.
    \end{proof}
    \begin{theorem}\label{thm:orderedapprox}
        Algorithm \algname{OrderedApprox} computes an $(\alpha + 2\gamma - 2)$-approximate independent set $I$ in $\mathcal{I}$ in polynomial time, where $\gamma$ is the width of a given graph $G$ with respect to a given linear order of vertices $V$.
    \end{theorem}
    \begin{proof}
        For any $v$ in $V$, let $I_v$, $P_v$, and $R_v$ be the sets obtained by Algorithm \algname{OrderedApprox}. By Lemma \ref{lem:independence}, $I = \bigcup_{v \in V} I_v$ is an independent set in $\mathcal{I}$. Let us first consider the size of $R_{v_i}$. If the $i$-th iteration of the for-loop falls into Case 1, then we have $R_{v_i} = \emptyset$. If it falls into Cases 2 or 3, then we have $|R_{v_i}| = |U_{v_i}| - 1 \leq \gamma - 1$ and $I_{v_i} \not = \emptyset$, which implies $|R_{v_i}| \leq (\gamma - 1)|I_{v_i}|$. In Case 4, $v_i$ is added to $X$ and $R_{v_i}$ is set to $R_{v_i} = \emptyset$ at the end of the $i$-th iteration. 
        Note that $R_{v_i}$ might be updated to $\{(v_i, v_k) \in U_{v_k} \mid k > j\}$ if $(v_i, v_j) \in U_{v_i}$ is chosen by $I_{v_j}$ in the $j$-th iteration for some $j > i$. In either case, we have $|R_{v_i}| \leq \gamma - 1$, and there exists an edge $(v_i, v_j)$ in $U_{v_i} \cap I_{v_j}$ if $R_{v_i} \not = \emptyset$. For $p=1,2,3$, and $4$, let $V_p$ denote the set of vertices $v_i$ such that the $i$-th iteration of the for-loop falls into Case $p$. Then we have 
        \begin{align*}
        \sum_{v \in V_1} |R_v| &= 0,\\
        \sum_{v \in V_2 \cup V_3} |R_v| &\leq (\gamma - 1)\sum_{v \in V_2 \cup V_3} |I_v| \leq (\gamma - 1)|I|,\ \text{and}\\
        \sum_{v \in V_4} |R_v| &\leq (\gamma - 1)|I|.
        \end{align*}
         Therefore, we obtain the following inequality 
        \begin{align*}
            |R| = \sum_{v \in V} |R_v| \leq 2(\gamma - 1)|I|.
        \end{align*}
        By Lemma \ref{lem:conditionP}, $\mathcal{P} = \{P_v \mid v \in V\}$ is a local subpartition of $E$ with the residual $R=\bigcup_{v \in V}R_{v}$. Thus, by applying Lemma \ref{lem:subpart} to this $I$, we can see that $I$ is an $(\alpha + 2\gamma -2)$-approximate independent set in $\mathcal{I}$.
    \end{proof}

    Since a linear order $\prec$ of vertices $V$ representing the degeneracy of $G=(V, E)$ can be computed in linear time, we have the following corollary.
    \begin{corollary}
        Algorithm \algname{OrderedApprox} computes an $(\alpha + 2k - 2)$-approximate independent set $I$ in $\mathcal{I}$ in polynomial time, if a given graph $G$ is $k$-degenerate.
    \end{corollary}

    For a graph $G=(V, E)$ with the width $\gamma$, the second algorithm, called \algname{DecomApprox}, first decomposes edge set $E$ into forests $E_1, \ldots ,E_\gamma$, for each $E_i$, applies \algname{OrderedApprox} to compute an independent set $I_i$, and chooses a maximum independent set among them.
    \begin{algorithm}[ht]
        \caption{\algname{DecomApprox}$(G, \mathcal{I}, \prec)$}
        \label{alg:repaet}
        \begin{algorithmic}
            \Require An independence system $(E, \mathcal{I})$ defined on a graph $G = (V, E)$ and a linear order $v_1 \prec v_2 \prec \dots \prec v_n$ of vertices $V$, where $\gamma$ is the width of graph $G$ with respect to $\prec$.
            \Ensure An independent set in $\mathcal{I}$.
            \vspace{5pt}
            \For{$i = 1, \dots, \gamma$}
                \State $E_i := \emptyset$.
            \EndFor
            \For{each $v$ in $V$}
                \For{each $e_i$ in $U_v = \{e_1, e_2, \dots, e_{|U_v|}\}$}
                    \State $E_i := E_i \cup \{e_i\}$.
                \EndFor
            \EndFor
            \For{$i = 1, \dots, \gamma$}
                \State $I_i = \algname{OrderedApprox}(G[E_i],\mathcal{I}[E_i], \prec)$.
            \EndFor
            \State $I \in \arg \max \{|I_i| \mid i = 1, \dots,\gamma\}$.
            \State Output $I$ and halt.
        \end{algorithmic}
    \end{algorithm}
    
    Note that $G_i = (V, E_i)$ is a $1$-degenerate for any $i = 1, \dots, \gamma$. Thus we have the following theorem.
    
    \begin{theorem}\label{thm:DecomApprox}
        Algorithm \algname{DecomApprox} computes an $\alpha \gamma$-approximate independent set $I$ in $\mathcal{I}$ in polynomial time, where $\gamma$ is the width of a given graph $G$ with respect to a given linear order of vertices $V$.
    \end{theorem}
    \begin{proof}
        Since $G_i = (V, E_i)$ is a forest (i.e., $G_i$ has the width $1$), by Theorem \ref{thm:orderedapprox}, $I_i$ is an $\alpha$-approximate independent set of $\mathcal{I}[E_i]$. Furthermore, any independent set $J$ in $\mathcal{I}$ satisfies that
        \begin{align*}
            |J| = \sum_{i=1}^\gamma |J \cap E_i| \leq \sum_{i=1}^\gamma \alpha|I_i| \leq \alpha\gamma\max_{i=1, \ldots, \gamma}|I_i|,
        \end{align*}
        which completes the proof.
    \end{proof}
    Note that $\alpha < 2$ if and only if $\alpha\gamma < \alpha + 2\gamma -2$. Therefore, an independent set $I$ provided by Algorithm \algname{DecomApprox} has approximation guarantee better than the one provided by Algorithm \algname{OrderedApprox} when $\alpha < 2$. By applying Theorem \ref{thm:DecomApprox} to graphs with degeneracy $k$, we have the following corollary.
    \begin{corollary}
        Algorithm \algname{DecomApprox} computes  an $\alpha k$-approximate independent set $I$ in $\mathcal{I}$ in polynomial time if a given graph $G$ is $k$-degenerate.
    \end{corollary}

\section{Approximation for bipartite graph}
    We note that Algorithms \algname{OrderedApprox} and \algname{DecomApprox} do not provide an independent set with {\it small} approximation ratio if a given graph has no {\it small} \begin{revision}{1} degeneracy\end{revision}. Such examples include complete graphs and complete bipartite graphs.

    In this section, we consider an approximation algorithm for the problem \eqref{problem} where the input graph is bipartite and its degeneracy might not be bounded, and analyze the approximation ratio of the algorithm if all the local independence systems in the one-side of vertices are $k$-systems with independence oracles. Namely, let $(E, \mathcal{I})$ be an independence system defined on a bipartite graph $G=(V_1 \cup V_2, E)$. We consider the case, where every $v \in V_2$ satisfies that $(E_v, \mathcal{I}_v)$ is a $k$-system. 

    \begin{definition}
        For a positive $k \in \mathbb{R}$, an independence system $(E, \mathcal{I})$ is called a {\em $k$-system} if any subset $F \subseteq E$ satisfies
        \begin{align}
            k|I| \geq |J|\ \text{for any two maximal independent sets}\ I\ \text{and}\ J\ \text{in}\ \mathcal{I}[F].
        \end{align}
        
    \end{definition}
    
    Note that any independence system $(E, \mathcal{I})$ is a $|E|$-system and that an independence system is a $1$-system if and only if it is a matroid. By definition, matchoids are independence systems such that local independence systems are all $1$-systems, and hence the families of $b$-matchings are also the ones satisfying that local independence systems are all $1$-systems. Moreover, the families of timed matchings are independence systems such that local independence systems $(E_v, \mathcal{I}_v)$ are all $k$-systems if any time label $L_e$ with $e \in E_v$ is disjoint from $L_f$ with $f \in E_v$ except for at most $k$ edges in $E_v$.
    \begin{algorithm}[ht]
        \caption{\algname{BipartiteApprox}$(G, \mathcal{I})$}
        \label{alg:biartite}
        \begin{algorithmic}[1]
            \Require An independence system $(E, \mathcal{I})$ defined on a bipartite graph $G = (V_1 \cup V_2, E)$, where $V_1 = \{v_1, \dots, v_{n_1}\}$ and $(E_v, \mathcal{I}_v)$ is a $k$-system with an independence oracle for every $v \in V_2$.
            \Ensure An independent set in $\mathcal{I}$.
            \vspace{5pt}
            \State $P_v:=E_v$ for $v \in V_1$.
            \State $R_v:=\emptyset$ and $J_v:=\emptyset$ for $v \in V_2$.
            \For{$i = 1, \ldots, n_1 (= |V_1|)$}\label{algline:firstfor-bi}
            	\State $I_{v_i} := \mathcal{A}_{v_i}(P_{v_i})$. \label{algline:compute_I}
				\For{$(w, v_i) \in I_{v_i}$}
				    \State $J_{w} := J_{w} \cup \{(w, v_i)\}$\label{algline:update_J}
				    \State $R_{w} := R_{w} \cup \{(w, v_j) \in E_{w} \mid j>i\ \text{and}\ J_{w} \cup \{(w, v_j)\} \not \in \mathcal{I}_{w} \}$.\label{algline:update_R-bi}
				\EndFor
				\For{$j = i+1, \ldots, n_1$}
				\State $P_{v_j} = P_{v_j} \setminus (\bigcup_{v \in V_2} R_{v})$.\label{algline:update_P-bi}
				\EndFor
            \EndFor
            \State $R = \bigcup_{v \in V_2} R_v$.
            \State Output $I = \bigcup_{v \in V_1} I_v$ and halt.
        \end{algorithmic}
    \end{algorithm}
    Our algorithm called \algname{BipartiteApprox} can be regarded as variant of Algorithm \algname{OrderedApprox} with a linear order $\prec$ such that $w \prec v$ hold for any $v \in V_1$ and $w \in V_2$. Note that in this order all the vertices $w \in V_2$ fall into Case $4$ at the for-loop in Line \ref{algline:firstfor} in \algname{OrderedApprox}, and hence $X = V_2$ holds after the iteration of the last vertex in $V_2$. Different from \algname{OrderedApprox}, Algorithm \algname{BipartiteApprox} updates $R_w$ for $w \in V_2\ (=X)$ more carefully.
    
    Algorithm \algname{BipartiteApprox} calls local oracles $\mathcal{A}_v$ in an arbitrary order of $v \in V_1$. More precisely, for each $v \in V_1$ and $w \in V_2$, we initialize $P_v = E_v$ and $R_w = \emptyset$, update $P_v$ and $R_w$ accordingly, and compute $I_v = \mathcal{A}_v(P_v)$. Here $P_v$ and $R = \bigcup_{w \in V_2} R_v$ correspond to those in Lemma \ref{lem:subpart}. 
    
    \begin{lemma}\label{lem:conditionP-bi}
        Algorithm \algname{BipartiteApprox} satisfies the following five conditions
        \begin{flalign*}
            \hspace{15pt} \rm{({\romannumeral 1})}&\ P_{v_j} \subseteq E_{v_j}\ \text{with}\ I_{v_j} = \mathcal{A}_{v_j}(P_{v_j})\ \text{for all }\ v_j \in V_1\ \text{with}\ j \leq i,&\\
            \rm{({\romannumeral 2})}&\ R_{v} \subseteq E_{v}\ \text{for all }\ v \in V_2,&\\
            \rm{({\romannumeral 3})}&\ \{J_v \subseteq E_v \mid v \in V_2\}\ \text{is a partition of}\ \bigcup_{j \leq i}I_{v_j},&\\
            \rm{({\romannumeral 4})}&\ J_{v} \ \text{is a maximal independent set in}\ \mathcal{I}[J_v \cup R_v]\ \text{for all}\ v \in V_2,\ \text{and}&\\
            \rm{({\romannumeral 5})}&\ \mathcal{P} = \{P_v \mid v \in V_1\} \cup \{P_w = \emptyset \mid w \in V_2\}\ \text{is a local subpartition of}\ E\\ &\ \text{with residual}\ R=\bigcup_{v \in V_2}R_{v}&
        \end{flalign*}
        at the end of the $i$-th iteration of the for-loop in Line \ref{algline:firstfor-bi}.
    \end{lemma}
        \begin{proof}
        Since $I_{v_j}$ and $P_{v_j}$ are never modified after the $j$-th iteration of the for-loop in Line \ref{algline:firstfor-bi}, it is enough to show that $P_{v_i} \subseteq E_{v_i}$ and $I_{v_i} = \mathcal{A}_{v_i}(P_{v_i})$ at the end of the $i$-th iteration to prove Condition \rm{({\romannumeral 1})}. We can see that $P_{v_i}$ is initialized to $E_{v_i}$, $P_{v_i}$ is modified without adding edge in Line \ref{algline:update_P-bi}, and $I_{v_i} = \mathcal{A}_{v_i}(P_{v_i})$ holds in Line \ref{algline:compute_I}. This implies the desired property.
        
        For $w \in V_2$, we initialize $J_w, R_w= \emptyset$, and they are respectively updated in Lines \ref{algline:update_J} and \ref{algline:update_R-bi} with adding edges in $E_w$, which satisfies $J_w, R_{w} \subseteq E_w$. This implies Condition \rm{({\romannumeral 2})}. Moreover, during the $j$-th iteration, any $(w, v_j) \in I_{v_j}$ is added to $J_{w}$, which implies Condition \rm{({\romannumeral 3})}. It is not difficult to see that Condition \rm{({\romannumeral 4})} is satisfied since $J_w$ and $R_w$ are updated in Lines \ref{algline:update_J} and \ref{algline:update_R-bi}.
        
        Let us finally show Condition \rm{({\romannumeral 5})}. Since $G$ is a bipartite graph, before the first iteration of the for-loop, $\{P_v = E_v \mid v \in V_1\}$ is a partition of $E$. Since $P_{v_j}$ is modified without adding edge in Line \ref{algline:update_P-bi}, $\{P_v \mid v\in V_1\}$ is a subpartition of $E$. Moreover, $R_{w}$ is initialized to $R_{w} = \emptyset$ for $w \in V_2$, and when some edge $e = (w, v_j)$ is added to $R_{w}$ in Line \ref{algline:update_R-bi}, all such edges $e$ are deleted from the corresponding $P_{v_j}$ in Line \ref{algline:update_P-bi}. Thus Condition \rm{({\romannumeral 5})} is satisfied at the end of the $i$-th iteration, which completes the proof of the lemma.
    \end{proof}

    \begin{lemma}\label{lem:independence-bi}
    Algorithm \algname{BipartiteApprox} satisfies that
        $\bigcup_{j < i}I_{v_j} \in \mathcal{I}$
    in the beginning of the $i$-th iteration of the for-loop in Line \ref{algline:firstfor-bi}.
    \end{lemma}
    \begin{proof}
        We prove the statement in Lemma \ref{lem:independence-bi} by the induction on $i$. If $i = 1$, the statement clearly holds since $\emptyset \in \mathcal{I}$. Assume that the statement holds when $i = l$ and consider the case $i = l+1$. 
        
        For $v \in V_1$ and $q = l, l+1$, let $P^{(q)}_v$ denotes the set $P_v$ in the beginning of the $q$-th iteration of the for-loop in Line \ref{algline:firstfor-bi}. Note that for $v \in V_1$, $I_v$ is never modified after it is computed in Line 4. Similarly, for $w \in V_2$ and $q = l, l+1$, let $J^{(q)}_w$ and $R^{(q)}_w$ respectively denote the sets $J_w$ and $R_w$ in the beginning of the $q$-th iteration, and let $J^{(q)} = \bigcup_{w \in V_2}J^{(q)}_w$.
        
        By Lemma \ref{lem:conditionP-bi} \rm{({\romannumeral 3})}, we have 
        \begin{align*}
        \bigcup_{j < l+1}I_{v_{j}} = I_{v_{l}} \cup \bigcup_{j < l}I_{v_{j}} = I_{v_{l}} \cup \bigcup_{w \in V_2}J^{(l)}_w. 
        \end{align*}
        
        Note that $J^{(l)}_{w} \cup \{e\} \in \mathcal{I}_{w}$ for any $w \in V_2$ and $e \in E_w \setminus R^{(l)}_w$. Since $I_{v_l} \subseteq P^{(l)}_{v_{l}} \subseteq E \setminus (\bigcup_{w \in V_2} R^{(l)}_w)$ by Lemma \ref{lem:conditionP-bi} \rm{({\romannumeral 1})} and \rm{({\romannumeral 5})}, 
        \begin{align*}
            I_{v_{l}} \cup \bigcup_{w \in V_2}J^{(l)}_w \in \mathcal{I},
        \end{align*}
        which completes the inductive argument.
    \end{proof}
    
    \begin{theorem}\label{thm:bipartiteapprox}
        Let $(E, \mathcal{I})$ be an independence system defined on a bipartite graph $G = (V_1 \cup V_2, E)$ such that $(E_v, \mathcal{I}_v)$ is a $k$-system with an independence oracle for every $v \in V_2$. Then Algorithm \algname{BipartiteApprox} computes an $(\alpha + k)$-approximate independent set $I$ in $\mathcal{I}$ in polynomial time.
    \end{theorem}
    \begin{proof}
        For a vertex $v \in V_1$, let $I_v$ and $P_v$ be the sets obtained by Algorithm \algname{BipartiteApprox}, and for a vertex $w \in V_2$, let $J_w$ and $R_w$ be the sets obtained by the algorithm. By Lemma \ref{lem:independence-bi}, $I = \bigcup_{v \in V_1} I_v$ is an independent set in $\mathcal{I}$. By Lemma \ref{lem:conditionP-bi}, $\mathcal{P} = \{P_v \mid v \in V_1\}$ is a local subpartition of $E$ with the residual $R=\bigcup_{w \in V_2}R_{w}$. 
        Let us analyze $\max_{L \in \mathcal{I}[R]}|L|/|I|$ to apply Lemma \ref{lem:subpart}. By Lemma \ref{lem:conditionP-bi} \rm{({\romannumeral 4})}, $J_w$ is a maximal independent set in $\mathcal{I}[J_w \cup R_w]$ for $w \in V_2$. Since the every local independence system $(E_w, \mathcal{I}_w)$ is a $k$-system for $w \in V_2$, 
        \begin{align*}
            |L| = \sum_{w \in V_2}|L \cap R_w| \leq k\sum_{w \in V_2}|J_w| = k|I|
        \end{align*}
        for any $L \in \mathcal{I}[R]$. Thus, by applying Lemma \ref{lem:subpart} to this $I$, we can conclude that $I$ is an $(\alpha + k)$-approximate independent set in $\mathcal{I}$.
    \end{proof}
    Before concluding this section, we remark that independence systems on graphs are $2k$-systems if all local independence systems are $k$-systems, which is shown in the following lemma.
    Since the maximization for $k$-systems $(E, \mathcal{I})$ are approximable with ratio $k$ by computing a maximal independent set in $\mathcal{I}$, 
    Algorithm \algname{BipartiteApprox} improves the ratio $2k$ to $(\alpha + k)$ for independence systems on graph bipartite graphs $G = (V_1\cup V_2, E)$ if the all local independence systems are $k$-systems and $\alpha$-approximation local oracles in $V_1$ are given for $\alpha$ with $\alpha < k$.
    \begin{lemma}
        An independence system $(E, \mathcal{I})$ on graph $G=(V, E)$ is a $2k$-system if all local independence systems $(E_v, \mathcal{I}_v)$ are $k$-systems.
    \end{lemma}
    \begin{proof}
        For a subset $F \subseteq E$, let $I$ be a maximal independent set in $\mathcal{I}[F]$. For each vertex $v\in V$, let $I_v = I \cap E_v$, and define $S_v$ and $T_v$ by
        \begin{align*}
            S_v &= \{(u, v) \in F \setminus I \mid I_v \cup \{(u, v)\} \in \mathcal{I}_v\}\\
            T_v &= \{(u, v) \in F \setminus I \mid I_u \cup \{(u, v)\} \in \mathcal{I}_u\}.
        \end{align*}
        Since $I$ is a maximal independent set in $\mathcal{I}[F]$, $S_v \cap T_v = \emptyset$ for all $v \in V$, and $\mathcal{S} = \{S_v \mid v \in V\}$ and $\mathcal{T} = \{T_v \mid v \in V\}$ are both subpartitions of $F$ with $\bigcup_{v \in V} S_v = \bigcup_{v \in V} T_v$. Moreover, $I_v$ is a maximal independent set in $\mathcal{I}[I_v \cup (E_v \setminus S_v)]$ and $\mathcal{I}[I_v \cup T_v]$. Therefore, for any independent set $J$ in $\mathcal{I}[F]$, we have
        \begin{align*}
            |J| &= \frac{1}{2} \sum_{v \in V}|J \cap E_v|\\
            &= \frac{1}{2}\sum_{v \in V}|J \cap S_v| + \frac{1}{2}\sum_{v \in V}|J \cap (E_v \setminus S_v)| \\
            &= \frac{1}{2} \sum_{v \in V}|J \cap T_v| + \frac{1}{2}\sum_{v \in V}|J \cap (E_v \setminus S_v)| \\
            &\leq \frac{1}{2}\sum_{v \in V}k|I_v| + \frac{1}{2}\sum_{v \in V}k|I_v|\\
            &= 2k|I|.
        \end{align*}
        Here the second equality follows from $\mathcal{S}$ is a subpartition of $F$, the third equality follows from $\bigcup_{v \in V} S_v = \bigcup_{v \in V} T_v$, and the first inequality follows from the fact that local independence systems are all $k$-systems.
    \end{proof}
    
\section{Hypergraph generalization of the probslem}
    In this section we consider independence systems defined on hypergraphs and present two algorithms. The first algorithm corresponds to Algorithm \algname{OrderedApprox} for the problem \eqref{problem} whose input graph is a forest, and the second algorithm corresponds to Algorithm \algname{DecomApprox} for the problem \eqref{problem}.

    \begin{definition}
        Let $G=(V, E)$ be a hypergraph with a vertex set $V$ and a hyperedge set $E \subseteq 2^V$.
        For each vertex $v$ in $V$, let $(E_v, \mathcal{I}_v)$ be an independence system on the set $E_v$ of hyperedges incident to $v$, i.e., $E_v = \{e \in E \mid v \in e \}$, and let $\mathcal{I} = \{I \subseteq E \mid I \cap E_v \in \mathcal{I}_v\ \text{for all}\ v \in V \}$. We say that $(E, \mathcal{I})$ {\em is an independence system defined on a hypergraph} $G$.
    \end{definition}
    The generalization contains the maximum matching problem for hypergraphs. Similarly to the graph case, we make use of local oracles $\mathcal{A}_v$ for each $v$ in $V$, which satisfy (\ref{ALO1}) and (\ref{ALO2}). In order to generalize the idea of Algorithm \algname{OrderedApprox} to the hypergraph case, we define the upward and downward edge sets with respect to a linear order of vertices $V$.
    \begin{definition}\label{def:horder}
        For a hypergraph $G = (V, E)$, let $\prec$ be a linear order of vertices $V$.
        For a vertex $v$, let $U_{v} = \{e \in E_v \mid v \prec w\ \text{for some}\ w \in e\}$ and $D_{v} = \{e \in E_v \mid w \prec v\ \text{for all}\ w \in e\}$ be the sets of {\em upward} and {\em downward} edges incident to $v$, respectively. We define the {\em width} of $G$ (with respect to a linear order $\prec$) as $\max_{v \in V}|U_{v}|$.
    \end{definition}
    For a hypergraph $G=(V, E)$ and a vertex set $W \subseteq V$, $G[W]=(W, E[W])$ denotes the simple subhypergraph of $G$ induced by $W$, where $E[W] = \{e \cap W \mid e \in E,\ |e\cap W| \geq 2\}$. Note that for $W \subseteq V$, $G[W]$ is the subgraph of $G$ induced by $W$ if $G$ is a graph. Similarly to the graph case, a hypergraph $G$ of width $k$ satisfies that $G[W]$ has a vertex of degree at most $k$ for all $W \subseteq V$. Therefore, a linear order certifying that a hypergraph $G$ has width $k$ can be computed in linear time.
    
    Algorithm \algname{OrderedApprox$^*$} for hypergraphs works similarly to Algorithm \algname{OrderedApprox}. Different from Algorithm \algname{OrderedApprox}, the algorithm updates $P_v$ and $R_v$ based on the fact that $G$ is a hypergraph. Note that such differences appear in Lines 4, 14, and 28 in the description of the algorithm.
    \begin{algorithm}
        \caption{\rm{\algname{OrderedApprox$^*$}}$(G, \mathcal{I}, \prec)$}
        \label{alg:hordered}
        \begin{algorithmic}[1]
            \Require An independence system $(E, \mathcal{I})$ defined on a hypergraph $G = (V, E)$ and a linear order $v_1 \prec v_2 \prec \ldots \prec v_n$ of vertices $V$.
            \Ensure An independent set in $\mathcal{I}$.
            \vspace{5pt}
            \State $X := \emptyset$.
            \State $P_v := D_v$ for $v \in V$.
            \For{$i = 1, \ldots, n$}\label{algline:hfirstfor}
                \State $B_{v_i} = \left\{e \in U_{v_i} \setminus \left(\bigcup_{l < i} R_{v_l}\right) \middle| \begin{array}{l}
                    e \in \mathcal{A}_{v_i}(P_{v_i} \cup \{e\})\ \text{and}\medskip\\
                    | \mathcal{A}_{v_i}(P_{v_i} \cup \{e\})| > |\mathcal{A}_{v_i}(P_{v_i})|
                \end{array}\right\}.$ 
            	\If{$U_{v_i}= \emptyset$}
    				\State $I_{v_i} := \mathcal{A}_{v_i}(P_{v_i})$.
    				\State $R_{v_i}:=\emptyset$. \label{algline:hconstructR-1}
            	\ElsIf{$U_{v_i} \not = \emptyset$, $P_{v_i} \not = \emptyset$, and $B_{v_i} = \emptyset$}
				    \State Choose an edge $e \in U_{v_i}$ arbitrarily.
				    \If{$e \not \in \mathcal{A}_{v_i}(P_{v_i} \cup\{e\})$}
				        \State $I_{v_i} := \mathcal{A}_{v_i}(P_{v_i} \cup \{e\})$. 
				    \Else{ (i.e., $|\mathcal{A}_{v_i}(P_{v_i} \cup \{e\})| = |\mathcal{A}_{v_i}(P_{v_i})|$)}
				        \State $I_{v_i} := \mathcal{A}_{v_i}(P_{v_i})$.\Comment{16.5em}{$|I_{v_i}| = |\mathcal{A}_{v_i}(P_{v_i}\cup \{e\})|$} 
				    \EndIf
    				\State $P_{v_i} := P_{v_i} \cup \{e\}$.\label{algline:hupdate_P-1}
    				\State $R_{v_i} := U_{v_i} \setminus \{e\}$. \label{algline:hconstructR-2}
    				\For{$v_j \in e \setminus \{v_i\}$}
    				    \State $P_{v_j} := P_{v_j} \setminus \{e\}$.\label{algline:hupdate_P-2}
    				\EndFor
    			\ElsIf{$U_{v_i} \not = \emptyset$, $P_{v_i} \not = \emptyset$, and $B_{v_i} \not = \emptyset$}
				    \State Choose an edge $b \in B_{v_i}$ arbitrarily.
					\State $I_{v_i} := \mathcal{A}_{v_i}(P_{v_i} \cup \{b\}) \setminus \{b\}$. 
					\Comment{17em}{$I_{v_i} \subseteq P_{v_i}$ and $|I_{v_i}| = |\mathcal{A}_{v_i}(P_{v_i})|$}
					\State $R_{v_i} := U_{v_i} \setminus \{b\}$.\label{algline:hconstructR-3}
    			\Else{ (i.e., $U_{v_i} \not = \emptyset$ and $P_{v_i} = \emptyset$)}
    				\State $I_{v_i} := \emptyset$.
    				\State $R_{v_i}:=\emptyset$. \label{algline:hconstructR-4} \Comment{17em}{$R_{v_i}$ might be updated in line \ref{algline:hupdate_R}}
    				\State $X := X \cup \{v_i\}$.
                \EndIf
				\algstore{break}
        \end{algorithmic}
    \end{algorithm}
    \begin{algorithm}
        \begin{algorithmic}[1]
            \algrestore{break}
            \For{$e \in I_{v_i}$}
				    \For{$v_j \in e \setminus \{v_i\}$}
				        \If{$v_j \in X$}
				            \State $R_{v_j} := \{e \in U_{v_j} \setminus (\bigcup_{l < i} R_{v_l}) \mid v_i \prec u\ \text{for some}\ u \in e\}$.\label{algline:hupdate_R}
				            \State $X := X \setminus \{v_j\}$.
				        \EndIf
				    \EndFor
				\EndFor
            \For{$j = i+1, \ldots, n$}
				\State $P_{v_j} := P_{v_j} \setminus (\bigcup_{l \leq i} R_{v_l})$.\label{algline:hupdate_P-3}
				\EndFor
            \EndFor
            \State $R= \bigcup_{v \in V} R_v$.
            \State Output $I= \bigcup_{v \in V} I_v$ and halt.
        \end{algorithmic}
    \end{algorithm}
    For each $v_i \in V$, define a set $B_{v_i} \subseteq U_{v_i}$ as follows
    \begin{align*}
        B_{v_i} = \left\{e \in U_{v_i} \setminus \left(\bigcup_{l < i} R_{v_l}\right) \middle| \begin{array}{l}
             e \in \mathcal{A}_{v_i}(P_{v_i} \cup \{e\})\ \text{and}\medskip\\
             | \mathcal{A}_{v_i}(P_{v_i} \cup \{e\})| > |\mathcal{A}_{v_i}(P_{v_i})|
        \end{array}
        \right\}.
    \end{align*} 
    It separately treats the following four cases as in the description of the algorithm.
    \begin{flalign*}
        \hspace{20pt}\text{Case 1:}&\ U_{v_i} = \emptyset &\\
        \text{Case 2:}&\ U_{v_i} \not= \emptyset,\ P_{v_i} \not= \emptyset,\ \text{and}\ B_{v_i} = \emptyset&\\
        \text{Case 3:}&\ U_{v_i} \not= \emptyset,\ P_{v_i} \not= \emptyset,\ \text{and}\ B_{v_i} \not= \emptyset&\\
        \text{Case 4:}&\ U_{v_i} \not= \emptyset\ \text{and}\ P_{v_i} = \emptyset&
    \end{flalign*}
    We show the following two lemmas, which respectively correspond to Lemmas \ref{lem:conditionP} and \ref{lem:independence} for the graph case.
    
    \begin{lemma}\label{lem:hconditionP}
    Algorithm \algname{OrderedApprox$^*$} satisfies the following three conditions
    \begin{flalign*}
        \hspace{15pt} \rm{({\romannumeral 1})}&\ I_{v_j} \subseteq P_{v_j} \subseteq E_{v_j}\ \text{with}\ |I_{v_j}| = |\mathcal{A}_{v_j}(P_{v_j})|\ \text{for all }\ j \leq i,&\\
        \rm{({\romannumeral 2})}&\ R_{v_j} \subseteq U_{v_j}\ \text{for all }\ j \leq i,\ \text{and}&\\
        \rm{({\romannumeral 3})}&\ \mathcal{P} = \{P_v \mid v \in V\}\ \text{is a local subpartition of}\ E\ \text{with residual}\ R=\bigcup_{j \leq i}R_{v_j}&
    \end{flalign*}
    at the end of the $i$-th iteration of the for-loop.
    \end{lemma}
    \begin{proof}
        It is analogous to the proof of Lemma \ref{lem:conditionP} for the graph case.
    \end{proof}
    Different from the graph case, the following lemma requires that $v \in V$ $e \cap f = \{v\}$ holds for $e, f \in P_{v}$, i.e, $\mathcal{A}_{v}(P_v) \in \mathcal{I}$. Note that if $\max_{e \in E} |e| \leq 2$, the required condition is satisfied.
    \begin{lemma}\label{lem:hindependence}
    Algorithm \algname{OrderedApprox$^*$} satisfies that in the beginning of the $i$-th iteration of the for-loop,
    \begin{align}
        I \cup \bigcup_{j < i}I_{v_j} \in \mathcal{I}\ \text{for any independent set}\ I \in \mathcal{I}[P_{v_i}] \label{lemeq:hindependence}
    \end{align}
    if $e \cap f = \{v_j\}$ holds for $j \leq i$ and for $e, f \in P_{v_j}$.
    \end{lemma}
    \begin{proof}
        For any $j < i$, consider the end of the $j$-th iteration of the for-loop in Line \ref{algline:hfirstfor}. We have $I_{v_j} \cup \{e\} \in \mathcal{I}_{v_j}$ for any $e \in U_{v_j} \setminus (P_{v_j} \cup R_{v_j})$. If the $j$-th iteration of the for-loop falls into Cases 1, 2, or 3, then $U_{v_j} \cap (\bigcup_{k > j} P_{v_k}) = U_{v_j} \setminus (P_{v_j} \cup R_{v_j})$ contains at most one edge. Otherwise we have $U_ {v_j} \subseteq (\bigcup_{k > j} P_{v_k}) \cup (\bigcup_{k < j} R_{v_k})$, and $R_{v_j}$ will be updated to $\{e \in U_{v_j} \setminus (\bigcup_{k < j} R_{v_k}) \mid v_k \prec v\ \text{for some}\ v \in e\}$ if $e \in U_{v_j}$ with $v_k \in e$ is chosen by $I_{v_k}$ in the $k$-th iteration for some $k > j$. Therefore, (\ref{lemeq:hindependence}) is satisfied in the beginning of the $i$-th iteration of the for-loop.
    \end{proof}
    \begin{lemma}\label{lem:horderedapprox}
        Algorithm \algname{OrderedApprox$^*$} computes an $(\alpha + \delta(\gamma - 1) )$-\!\! approximate independent set $I$ in $\mathcal{I}$ in polynomial time if a given hypergraph $G$ satisfies that for $v \in V$, $e \cap f = \{v\}$ for $e, f \in D_v$ with respect to a given linear order $\prec$ of vertices $V$, where $\gamma$ is the width of $G$ with respect to $\prec$, and $\delta = \max_{e \in E} |e|$.
    \end{lemma}
    \begin{proof}
        For any $v$ in $V$, let $I_v$, $P_v$, and $R_v$ be the sets obtained by Algorithm \algname{OrderedApprox$^*$}. By Lemma \ref{lem:hindependence}, $I = \bigcup_{v \in V} I_v$ is an independent set in $\mathcal{I}$. Let us first consider the size of $R_{v_i}$. If the $i$-th iteration of the for-loop falls into Case 1, then we have $R_{v_i} = \emptyset$. If it falls into Cases 2 or 3, then we have $|R_{v_i}| = |U_{v_i}| - 1 \leq \gamma - 1$ and $I_{v_i} \not = \emptyset$, which implies $|R_{v_i}| \leq (\gamma - 1)|I_{v_i}|$. Otherwise, $v_i$ is added to $X$ and $R_{v_i}$ is set to $R_{v_i} = \emptyset$ at the end of the $i$-th iteration. 
        Note that $R_{v_i}$ might be updated to $\{e \in U_{v_i} \setminus (\bigcup_{l < j} R_{v_l}) \mid v_j \prec u\ \text{for some}\ u \in e\}$ if $e \in U_{v_i}$ is chosen by $I_{v_j}$ in the $j$-th iteration for some $j > i$. In either case, we have $|R_{v_i}| \leq \gamma - 1$, and there exists an edge $e$ in $U_{v_i} \cap I_{v_j}$ if $R_{v_i} \not = \emptyset$. For $p=1,2,3$, and $4$, let $V_p$ denote the set of vertices $v_i$ such that the $i$-th iteration of the for-loop falls into Case $p$. Then we have 
        \begin{align*}
        \sum_{v \in V_1} |R_v| &= 0,\\
        \sum_{v \in V_2 \cup V_3} |R_v| &\leq (\gamma - 1)\sum_{v \in V_2 \cup V_3} |I_v| \leq (\gamma - 1)|I|,\ \text{and}\\
        \sum_{v \in V_4} |R_v| &\leq (\delta -1)(\gamma - 1)|I|. 
        \end{align*}
        Therefore, we obtain the following inequality
        \begin{align*}
            |R| = \sum_{v \in V} |R_v| \leq \delta(\gamma - 1)|I|.
        \end{align*}
        By Lemma \ref{lem:hconditionP}, $\mathcal{P} = \{P_v \mid v \in V\}$ is a local subpartition of $E$ with the residual $R=\bigcup_{v \in V}R_{v}$. Note that Lemma \ref{lem:subpart} still holds even for the hypergraph case. Thus, we can see that $I$ is an $(\alpha + \delta(\gamma -1))$-approximate independent set in $\mathcal{I}$.
    \end{proof}
    We have the following corollary.
    \begin{corollary}\label{cor:horderedapprox-forest}
        Algorithm \algname{OrderedApprox$^*$} computes an $\alpha$-approximate independent set $I$ in $\mathcal{I}$ in polynomial time if $G$ is a hypergraph of width $1$.
    \end{corollary}
     
    We can also show that Algorithm \algname{DecomApprox} can be generalized for the hypergraph case. Namely, for a hypergraph $G=(V, E)$ of width $\gamma$, Algorithm \algname{DecomApprox$^*$} first decomposes edge set $E$ into $E_1, \ldots ,E_\gamma$, where $G_i = (V, E_i)$ is a hypergraph of width $1$, for each $E_i$, applies \algname{OrderedApprox$^*$} to compute an independent set $I_i$, and chooses a maximum independent set among them.
    \begin{algorithm}[ht]
        \caption{\algname{DecomApprox$^*$}$(G, \mathcal{I}, \prec)$}
        \label{alg:hrepaet}
        \begin{algorithmic}[1]
            \Require An independence system $(E, \mathcal{I})$ defined on a hypergraph $G = (V, E)$ and a linear order $v_1 \prec v_2 \prec \dots \prec v_n$ of vertices $V$, where $\gamma$ is the width of hypergraph $G$ with respect to $\prec$.
            \Ensure An independent set in $\mathcal{I}$
            \vspace{5pt}
            \State $\mathcal{Q} := \emptyset$
            \For{$i = 1, \dots, n$}
                \State $U_{v_i}^\prime := \{e \in U_{v_i} \mid v_i \prec v\ \text{for all}\ v \in e \setminus \{v_i\}\}$.
                \For{each $e \in U_{v_i}^\prime$}
                    \State $u := \max_{v \in e}v$.
                    \State $\mathcal{Q}_{e} := \{Q^\prime \in \mathcal{Q} \mid Q^\prime \cap U_v = \emptyset\ \text{for all}\ v \in e \setminus \{u\}\}$.
                    \If{$\mathcal{Q}_e \not= \emptyset$}
                        \State $Q \in \mathcal{Q}_e$
                    \Else
                        \State $Q := \emptyset$.
                        \State $\mathcal{Q} := \mathcal{Q} \cup \{Q\} $.
                    \EndIf
                    \State $Q := Q \cup \{e\}$.
                \EndFor
            \EndFor
            \For{$Q \in \mathcal{Q}$}
                \State $I_Q :=d \algname{OrderedApprox$^*$}(G[Q],\mathcal{I}[Q], \prec)$.
            \EndFor
            \State $I \in \arg \max \{|I_Q| \mid Q \in \mathcal{Q}\}$.
            \State Output $I$ and halt.
        \end{algorithmic}
    \end{algorithm}

    Thus we have the following theorem.
    \begin{theorem}\label{thm:hDecomApprox}
        Algorithm \algname{DecomApprox$^*$} computes an $(\alpha + \alpha(\delta -1)(\gamma-1))$-approximate independent set $I$ in $\mathcal{I}$ in polynomial time, where $\gamma$ is the width of a given graph $G$ with respect to a given linear order of vertices $V$, and $\delta = \max_{e \in E} |e|$.
    \end{theorem}
    \begin{proof}
        Let $\mathcal{Q}$ and $\{U_{v}^\prime \mid v \in V\}$ be families of subsets of $E$ obtained by the algorithm. Note that $\mathcal{Q}$ is a partition of $E$ since for any $v \in V$ and $e \in U_{v}^\prime$, $e \in Q$ for some $Q \in \mathcal{Q}$, and $\{U_{v}^\prime \mid v \in V\}$ is a partition of $E$. 
        Moreover, for any $Q \in \mathcal{Q}$, $G[Q] = (V, Q)$ is a hypergraph of width $1$, i.e., $|Q \cap U_v| \leq 1$ holds for any $v \in V$, since for any $e \in Q$ and for any $v \in e \setminus \{\max_{u \in e} u\}$, $(Q \setminus \{e\}) \cap U_v = \emptyset$ holds. By Corollary \ref{cor:horderedapprox-forest}, $I_Q$ is an $\alpha$-approximate independent set of $\mathcal{I}[Q]$ for every $Q \in \mathcal{Q}$. Furthermore, any independent set $J$
        in $\mathcal{I}$ satisfies that
        \begin{align*}
            |J| = \sum_{Q \in \mathcal{Q}} |J \cap Q| \leq \sum_{Q \in \mathcal{Q}} \alpha|I_Q| \leq \alpha|\mathcal{Q}|\max_{Q \in \mathcal{Q}}|I_Q| = \alpha|\mathcal{Q}||I|.
        \end{align*}
        For any  $e \in E$, let $W_e = \{w \in V \setminus \{\min_{u \in e} u\} \mid e \in U_w\}$. Since $\delta = \max_{e \in E} |e|$ and $\gamma = \max_{v \in V} |U_v|$ hold, we have $|W_e| \leq \delta -2$ and $|U_w \setminus \{e\}| \leq \gamma -1$ for any $w \in W_e$. Then we have
        \begin{align*}
            |\mathcal{Q}| \leq (\delta - 2)(\gamma -1) + \gamma = (\delta -1)(\gamma -1) + 1, 
        \end{align*}
        which completes the proof.
    \end{proof}
    We also have the following corollary.
    \begin{corollary}
        Algorithm \algname{DecomApprox$^*$} computes an $(\alpha + \alpha(\delta -1)(k-1))$-approximate independent set $I$ in $\mathcal{I}$ in polynomial time if $G$ is a hypergraph of width $k$, where $\delta = \max_{e \in E} |e|$.
    \end{corollary} 
    
    \section*{Acknowledgments}
    I give special gratitude to Kazuhisa Makino for his valuable discussions and carefully proofreading of the manuscript. I am also grateful to Yusuke Kobayashi, Akitoshi Kawamura, and Satoru Fujishige for constructive suggestions.
    This work was supported by the joint project of Kyoto University and Toyota Motor Corporation, titled "Advanced Mathematical Science for Mobility Society".

\bibliographystyle{elsarticle-num}
\bibliography{IS}

\begin{thebibliography}{10}
\expandafter\ifx\csname url\endcsname\relax
  \def\url#1{\texttt{#1}}\fi
\expandafter\ifx\csname urlprefix\endcsname\relax\def\urlprefix{URL }\fi
\expandafter\ifx\csname href\endcsname\relax
  \def\href#1#2{#2} \def\path#1{#1}\fi

\bibitem{welsh1976matroid}
D.~Welsh, L.~M. Society, Matroid Theory, L.M.S. monographs, Academic Press,
  1976.

\bibitem{cook1997combinatorial}
W.~Cook, W.~Cook, W.~Cunningham, W.~Pulleyblank, A.~Schrijver, Combinatorial
  Optimization, A Wiley-Interscience publication, Wiley, 1997.

\bibitem{schrijver2003combinatorial}
A.~Schrijver, et~al., Combinatorial optimization: polyhedra and efficiency,
  Vol.~24, Springer, 2003.

\bibitem{oxley2006matroid}
J.~G. Oxley, Matroid theory, Vol.~3, Oxford University Press, USA, 2006.

\bibitem{lovasz1986matching}
M.~D. Plummer, L.~Lov{\'a}sz, Matching theory, Elsevier Science Ltd., 1986.

\bibitem{jenkyns1975matchoids}
T.~Jenkyns, Matchoids: a generalization of matchings and matroids., Ph.D.
  thesis, University of Waterloo (1975).

\bibitem{lovasz1978matroid}
L.~Lov{\'a}sz, The matroid matching problem, Algebraic methods in graph theory
  2 (1978) 495--517.

\bibitem{fujito19932}
T.~Fujito, A $2/3$-approximation of the matroid matching, in: Algorithms and
  Computation: 4th International Symposium, ISAAC'93, Hong Kong, December
  15-17, 1993. Proceedings, Vol. 762, Springer Science \& Business Media, 1993,
  p. 185.

\bibitem{lee2013matroid}
J.~Lee, M.~Sviridenko, J.~Vondr{\'a}k, Matroid matching: the power of local
  search, SIAM Journal on Computing 42~(1) (2013) 357--379.

\bibitem{avidorimproved}
A.~Avidor, I.~Berkovitch, U.~Zwick, Improved approximation algorithms for max
  nae-sat and max sat, Approximation and Online Algorithms (2005) 27--40.

\bibitem{kostakos2009temporal}
V.~Kostakos, Temporal graphs, Physica A: Statistical Mechanics and its
  Applications 388~(6) (2009) 1007--1023.

\bibitem{mandal20200}
S.~Mandal, A.~Gupta, 0-1 timed matching in bipartite temporal graphs, in:
  Conference on Algorithms and Discrete Applied Mathematics, Springer, 2020,
  pp. 331--346.

\bibitem{chekuri2004maximum}
C.~Chekuri, A.~Kumar, Maximum coverage problem with group budget constraints
  and applications, in: Approximation, Randomization, and Combinatorial
  Optimization. Algorithms and Techniques, Springer, 2004, pp. 72--83.

\bibitem{lick_white_1970}
D.~R. Lick, A.~T. White, $k$-degenerate graphs, Canadian Journal of Mathematics
  22~(5) (1970) 1082–1096.
\newblock \href {https://doi.org/10.4153/CJM-1970-125-1}
  {\path{doi:10.4153/CJM-1970-125-1}}.

\bibitem{freuder1982sufficient}
E.~C. Freuder, A sufficient condition for backtrack-free search, Journal of the
  ACM (JACM) 29~(1) (1982) 24--32.

\bibitem{matula1983smallest}
D.~W. Matula, L.~L. Beck, Smallest-last ordering and clustering and graph
  coloring algorithms, Journal of the ACM (JACM) 30~(3) (1983) 417--427.

\end{thebibliography}
\end{document}